\documentclass[onecolumn,journal]{IEEEtran}
\usepackage{indentfirst}
\usepackage{amsmath,amsthm,amssymb}
\usepackage{graphicx,multirow}

\newtheorem{Theorem}{Theorem}

\newtheorem{Example}{Example}

\newtheorem{Proposition}{Proposition}
\newtheorem{Property}{Property}

\newtheorem{Lemma}{Lemma}

\begin{document}

\title
{Optimal Exact Repair Strategy for the Parity Nodes of the $(k+2,k)$ Zigzag Code
\author{Jie Li and Xiaohu Tang, \IEEEmembership{Member,~IEEE}}
\thanks{J. Li is with the Information Security and National Computing Grid Laboratory, Southwest Jiaotong University, Chengdu, 610031, China (e-mail: jieli873@gmail.com).}
\thanks{X.H. Tang is with the Information Security and National Computing Grid
Laboratory, Southwest Jiaotong University, Chengdu 610031, China, and also
with the Beijing Center for Mathematics and Information Interdisciplinary
Sciences, Beijing 100048, China (e-mail: xhutang@swjtu.edu.cn).}
}

\date{}
\maketitle

\begin{abstract}
In this paper, we reinterprets the $(k+2,k)$ Zigzag code in coding matrix and then propose an optimal exact repair strategy for its parity nodes, whose repair disk I/O approaches a lower bound derived  in this paper.
\end{abstract}


\section{Introduction}

Distributed storage systems built on huge numbers of storage nodes have wide
applications in  peer-to-peer storage systems
such as OceanStore \cite{ocean}, Total Recall \cite{total} and DHash++
\cite{Dhash}. Erasure code, which can provide both protection against node failures and efficient data storage, is very common in  distributed storage systems
\cite{evenodd,Bcode,RDP,STAR,Bruck99,Xcode}.
For instance, as a special class of erasure code,   RAID-6 is a popular scheme for tolerating any two node failures \cite{RAID}.


Upon failure of a single node, a self-sustaining system should repair the failed node in order to retain the same redundancy. In the literature, there are mainly two repair types: exact repair and functional repair.
Compared with the latter, exact repair is preferred since it
does not incur additional significant system overhead by
regenerating the exact replicas of the lost data at the failed node \cite{survey}.
Generally speaking, there are several metrics to evaluate the performance of node repair, such as the \emph{repair bandwidth}, which is defined as the amount of data downloaded
from surviving nodes to repair a failed node, the \emph{disk I/O}, which is defined as the amount
of data  read.

Recently, Dimakis \emph{et al.} \cite{Dimakis}  introduced a new class of erasure code for distributed storage systems named \emph{minimum storage regenerating} (MSR) code.
The distributed storage system  deploys a $(k+r,k)$ MSR code  to store  a file of size $M = kN$ symbols across  $n$ nodes, each node keeping $N$ symbols.
The $(k+r,k)$ MSR code has the \emph{optimal repair property} that
the repair bandwidth $\gamma=\frac{d}{(d-k+1)}N$ is minimal, which is achieved  by downloading $N\over d-k+1$ symbols from each of any $k\le d\le k+r-1$ surviving nodes when repairing a failed node. In this paper, we only focus on the exact repair of  high rate MSR codes.
When $r=1$,  the repair bandwidth is the highest, i.e., $\gamma=M$. When $r=2$ and $d=k+1$, MSR code is very desirable since it can achieve the highest rate $k\over k+2$ for $\gamma=(k+1)N/2<M$. In addition, $(k+2,k)$ MSR code
can be alternative to RAID-6 schemes.

So far, several  explicit constructions of $(k+2,k)$ MSR codes have been presented \cite{invariant subspace,hadamard,Zigzag,Long}. Among them, the $(k+2,k)$ Zigzag code in \cite{Zigzag}, which is defined by a series of permutations, is of great interest because of:
\begin{enumerate}
  \item [(i)] Optimal update disk I/O property (also known as optimal update property in \cite{Zigzag}) that only  itself and one symbol at each parity node need an update  when a symbol in a systematic node is rewritten;
  \item [(ii)]  Optimal repair disk I/O property (also known as optimal rebuilding in \cite{Zigzag}) for systematic nodes that the repair disk I/O of a  systematic node is equal to the minimal repair bandwidth;
  \item [(iii)] Small alphabet size of $3$ so that it can be easily implemented;
  \item [(iv)] The storage $N=2^{k-1}$ achieves the theoretic lower bound on the storage per node  for $(k+2,k)$ MSR codes with both optimal update disk I/O and optimal repair disk I/O for systematic nodes \cite{Zigzag}.
\end{enumerate}
However, the parity nodes of the $(k+2,k)$ Zigzag code was trivially repaired by downloading all the original data in \cite{Zigzag}, i.e., the download bandwidth reaches
the maximal value $\gamma=M$.  In order to
acquire the optimal repair property for both systematic nodes and parity nodes, a $(k,k-2)$ MSR code was presented in  \cite{extend zigzag} based on a modification of the $(k+2,k)$ Zigzag code, but at cost of sacrificing two systematic nodes while maintaining the same storage per node $N=2^{k-1}$. It should be noted that only the $(k+2,k)$ Hadamard MSR code in \cite{hadamard}  shares the  optimally repair property of all the nodes in the all aforementioned codes.

In this paper, without changing the original structure of the $(k+2,k)$ Zigzag code, we propose an optimal repair strategy for the two parity nodes, whose download bandwidth achieves
the minimal value $\gamma=(k+1)N/2$. A comparison of the properties of various known $(k+2,k)$ MSR codes, such as the Zizag code employing our repair strategy,
the original Zigzag code \cite{Zigzag}, the modified Zigzag code \cite{extend zigzag}, and Hadamard code \cite{hadamard},
is given in Table \ref{table para compare}. It is seen that the new repair strategy does not lose any good properties of the original Zigzag code, for examples,
the optimal update disk I/O property, the optimal repair disk I/O property for systematic nodes, small alphabet size of $3$, and so on.
In contrast to the modified Zigzag code and Hadamard code
with the same optimal repair property of all nodes, the Zigzag code employing  the new  repair strategy
shows a clear advantage over the storage per node. Although the repair disk I/O of the parity node is not optimal, which is
$kN+N-k$, larger than the minimal repair bandwidth $(k+1)N/2$, it indeed approaches a lower bound on the disk I/O of Zigzag
code given in this paper.

\begin{table*}[htbp]\label{table para compare}
\centering
\caption{Comparison of the properties of some $(k+2,k)$ MSR codes where $q$ and $N$ denote the size of the finite field required and the storage per node, respectively.
}
\begin{tabular}{|c|c|c|c|c|c|c|c|}
\hline &&&&\multicolumn{2}{c|}{Optimal Repair Disk I/O}&\multicolumn{2}{c|}{Optimal Repair}\\
\cline{5-8}
&$q$& $N$ & Optimal Update Disk I/O & Systematic& Parity & Systematic &Parity\\
&&& & Nodes& Nodes & Nodes & Nodes\\
\hline Zizag Code &\multirow{2}{*}{$3$}&\multirow{2}{*}{$2^{k-1}$}&\multirow{2}{*}{Yes}&\multirow{2}{*}{Yes}&\multirow{2}{*}{No}&\multirow{2}{*}{Yes}&\multirow{2}{*}{Yes}\\
Employing  New Repair Strategy &&&&&&&\\
\hline Original Zigzag Code \cite{Zigzag}&$3$&$2^{k-1}$&Yes&Yes&No&Yes&No\\
\hline Modified Zigzag Code \cite{extend zigzag}&$3$&$2^{k+1}$&No &Yes&Yes&Yes&Yes\\
\hline Hadamard Code \cite{hadamard}&$2k+3$&$2^{k+1}$&Yes &No&No&Yes &Yes \\
\hline
\end{tabular}
\end{table*}

The rest of this paper is organized as follows. Section II introduces the structure of a $(k+2,k)$ MSR code and the necessary and sufficient conditions for optimal repair of parity nodes. Section III proposes the $(k+2,k)$ Zigzag code and reinterprets it in coding matrix. In Section IV, a lower bound on disk I/O to optimally repair the parity nodes of the $(k+2,k)$ Zigzag code is presented. The optimal repair strategy for the parity nodes of the $(k+2,k)$ Zigzag code is given in Section V.


\section{Optimal repair for parity nodes of $(k+2,k)$ MSR codes}
Let $q$ be a prime power and $\mathbf{F}_q$ be the finite field with $q$ elements.
Assume that a file of size $M=kN$ is equally partitioned into $k$ parts, respectively denoted by $\mathbf{f}_0, \mathbf{f}_1 ,\ldots, \mathbf{f}_{k-1}$, where $\mathbf{f}_j$ is a column vector of length $N$ for $0\le j<k$. The file is encoded to a $(k+2,k)$  MSR code and then stored across $k$ systematic and
two parity storage nodes, each node having storage $N$. The first $k$  nodes are systematic nodes, which store the file parts $\mathbf{f}_0,\mathbf{f}_1,\cdots,\mathbf{f}_{k-1}$ in an uncoded form respectively. Without loss of generality, assume that
the two parity nodes,  nodes $k$ and $k+1$,  respectively store $\mathbf{f}_k=\mathbf{f}_0+\mathbf{f}_1+\cdots+\mathbf{f}_{k-1}$  and $\mathbf{f}_{k+1}=A_0\mathbf{f}_0+A_1\mathbf{f}_1+\cdots+A_{k-1}\mathbf{f}_{k-1}$ for some $N\times N$ matrices $A_{0},\cdots,A_{k-1}$ over $\mathbf{F}_q$, where the matrix $A_{j}$ is called the \emph{coding matrix} for systematic node $j$, $0\le j< k$. To guarantee the MDS property, it is required that \cite{hadamard,new repair h}
\begin{eqnarray}\label{Eqn_MDS}
\mathrm{rank}(A_i)=\mathrm{rank}(A_i-A_j)=N, 0\le i\ne j<k.
\end{eqnarray}

Table I illustrates the structure of a $(k+2,k)$ MSR code.
\begin{table*}[htbp]\label{MSR_Model}
\begin{center}
\caption{Structure of a $(k+2,k)$ MSR code}
\begin{tabular}{|c|c|c|c|c|c|}
\hline Node 0 & Node 1 &$\cdots$ & Node $k-1$ &Node $k$ &Node $k+1$  \\
\hline $\mathbf{f}_0$ & $\mathbf{f}_1$ &$\cdots$ & $\mathbf{f}_{k-1}$ &$\mathbf{f}_k=\sum\limits_{i=0}^{k-1}\mathbf{f}_i$ & $\mathbf{f}_{k+1}=\sum\limits_{i=0}^{k-1}A_i\mathbf{f}_i$  \\
\hline
\end{tabular}
\end{center}
\end{table*}

When repairing a failed node $j$, the optimal repair property demands to download
half data from each surviving node $l$,  $0\le l\ne j< k+2$,  by
multiplying its original data $\mathbf{f}_l$ with an $N/2\times N$ matrix of rank $N/2$, called \emph{repair matrix}.
In what follows, we review the requirement on repair matrices for the optimal repair of parity nodes of a $(k+2,k)$ MSR code
\cite{hadamard, new repair h}.

Upon failure of the first parity  node (node $k$), respectively downloading $S_a\mathbf{f}_j$ and
$\tilde{S}_a\mathbf{f}_{k+1}$, $0\le j< k$, where  $S_a$ and $\tilde{S}_a$ are two $N/2\times N$  repair matrices of rank $N/2$, eventually one gets the following system of linear equations
\begin{eqnarray*}\label{Eqn download p1}
\left(
\begin{array}{c}
S_af_0 \\
\tilde{S}_af_{k+1}
\end{array}
\right)=\underbrace{\left(
\begin{array}{c}
S_a \\
\tilde{S}_aA_{0}
\end{array}
\right)\mathbf{f}_{k}}_{\mathrm{useful~ data}} - \sum_{l=1}^{k-1} \underbrace{\left( \begin{array}{c}
S_a \\
\tilde{S}_a (A_{0}-A_{l})
\end{array}
\right)\mathbf{f}_l}_{\mathrm{interference~ by}~\mathbf{f}_{l}}.
\end{eqnarray*}
To cancel all the interference terms and then recover the target data $\mathbf{f}_k$, the optimal repair requires \cite{hadamard,new repair h}
\begin{eqnarray}\label{repair_1parity_node_requirement1}
\textrm{rank} \left(\left(
\begin{array}{c}
S_a \\
\tilde{S}_aA_0
\end{array}
\right)\right) = N
\end{eqnarray}
and
\begin{eqnarray}\label{repair_1parity_node_requirement2}
\textrm{rank} \left(\left(
\begin{array}{c}
S_a \\
\tilde{S}_a (A_0- A_l)
\end{array}
\right)\right) = {N\over 2},\ \ 1\le l< k.
\end{eqnarray}
Clearly, the disk I/O to optimally repair the first parity node is $kN_1+N_2$ where $N_1$ and $N_2$
denote the  nonzero columns of $S_a$ and $\tilde{S}_a$ respectively.

To repair the  second parity node (node $k+1$),  downloading $(S_b A_j)\mathbf{f}_j$ and
$\tilde{S}_b\mathbf{f}_{k}$,  $0\le j< k$,  where $S_b$ and $\tilde{S}_b$ are two $N/2\times N$ matrices of rank $N/2$, one obtains the following system of linear equations
\begin{eqnarray*}\label{Eqn_download-1}
\left(
\begin{array}{c}
S_b A_0\mathbf{f}_{0}\\
\tilde{S}_b\mathbf{f}_{k}
\end{array}
\right)=\underbrace{\left(
\begin{array}{c}
S_b \\
\tilde{S}_bA_0^{-1}
\end{array}
\right)\mathbf{f}_{k+1}}_{\mathrm{useful~ data}} - \sum_{l=1}^{k-1} \underbrace{\left( \begin{array}{c}
S_b \\
\tilde{S}_b (A_0^{-1}-A_l^{-1})
\end{array}
\right)A_l\mathbf{f}_l}_{\mathrm{interference~ by}~\mathbf{f}_{l}}.
\end{eqnarray*}
Similarly,  optimal repair  demands \cite{hadamard,new repair h}
\begin{eqnarray}\label{repair_2parity_node_requirement1}
\textrm{rank} \left(\left(
\begin{array}{c}
S_b \\
\tilde{S}_bA_0^{-1}
\end{array}
\right)\right) = N
\end{eqnarray}
and
\begin{eqnarray}\label{repair_2parity_node_requirement2}
\textrm{rank} \left(\left(
\begin{array}{c}
S_b \\
\tilde{S}_b (A_0^{-1}-A_l^{-1})
\end{array}
\right)\right) = {N\over 2},\ \ 1\le l< k.
\end{eqnarray}
Accordingly, the disk I/O to optimally repair the second parity node is the total number of  nonzero columns of $\tilde{S}_b$ and $S_bA_i$, $0\le i< k$.

\section{Reinterpretation of $(k+2,k)$ Zigzag code in coding matrix}

Throughout this paper, let $k\ge 2$ and $N=2^{k-1}$. Given an integer $0\le i<N$, let $(i_1,\cdots,i_{k-1})$
be its binary expansion, i.e., $i=\sum\limits_{j=1}^{k-1}2^{k-1-j}i_{j}$. For simplicity, we do not distinguish a nonnegative integer $i$ and its binary expansion if the context is clear.

Let $\{e_j\}_{j=1}^{k-1}$ be the standard vector basis over $\mathbf{F}_2$ of dimension $k-1$, i.e.,
\begin{equation*}
    e_j=(\underbrace{0,\cdots,0,1,0,\cdots,0}\limits_{k-1}),\ \ 1\le j< k
\end{equation*}
with only the $j$th entry being nonzero. By convenience, set $e_0$ to be the all-zero vector.

In \cite{Zigzag}, the $(k+2,k)$ Zigzag code is characterized by the following permutation $P_j:\ [0,N-1]\rightarrow [0,N-1]$
\begin{equation*}\label{permutation}
P_j(x)=x\oplus e_j=\left\{
\begin{array}{cl}
(x_1,\cdots,x_{k-1}),& j=0\\
(x_1,\cdots,x_{j-1},x_{j}\oplus 1,x_{j+1},\cdots,x_{k-1}),& 0< j< k
\end{array}
\right.
\end{equation*}
where $\oplus$ denotes the addition in $\mathbf{F}_2$.
Obviously,
\begin{equation}\label{Eqn P and P-1}
  P_j^{-1}(x)=x\oplus e_j=P_j(x), \ 0\le j< k.
\end{equation}

For any integer $0\le l<N$, define $Z_l$ as $ Z_l=\{(i,j)|i=P_j^{-1}(l), 0\le j< k\}$,  i.e.,
\begin{equation*}\label{parity2}
  Z_l=\{(i,j)|i=l\oplus e_j, 0\le j< k\}
\end{equation*}
by \eqref{Eqn P and P-1}.
The structure of  the $(k+2,k)$ Zigzag code is depicted in Table II, where
the first parity node stores $f_{i,k}=\sum\limits_{j=0}^{k-1}f_{i,j}$ and the second parity node stores $f_{i,k+1}=\sum\limits_{(i,j)\in Z_l}\beta_{i,j} f_{i,j}$, $0\le i<N$
and $0\le j< k$,  $\beta_{i,j}=(-1)^{i\cdot \sum_{l=0}^{j}e_l}$, i.e.,
\begin{equation}\label{beta}
\beta_{i,j}=\left\{
\begin{array}{cl}
1, & \textrm{if} ~j=0\\
(-1)^{i_1+\cdots+i_j}, & \textrm{otherwise}
\end{array}
\right.
\end{equation}

In the following, we  reinterpret the data stored at the second parity node of the $(k+2,k)$ Zigzag code in the form of coding matrix
so that we can use Equations \eqref{repair_1parity_node_requirement1}-\eqref{repair_2parity_node_requirement2} to check the optimality of our new
repair matrices in the next section.

\begin{table*}[htbp]\label{Zigzag_Model}
\begin{center}
\caption{Structure of the $(k+2,k)$ Zigzag code}
\begin{tabular}{|c|c|c|c|c|}
\hline Node 0  &$\cdots$ & Node $k-1$ &Node $k$ &Node $k+1$  \\
\hline $f_{0,0}$ &$\cdots$ & $f_{0,k-1}$ &$f_{0,k}=\sum\limits_{j=0}^{k-1}f_{0,j}$ & $f_{0,k+1}=\sum\limits_{(i,j)\in Z_0}\beta_{i,j} f_{i,j}$  \\
\hline $f_{1,0}$  &$\cdots$ & $f_{1,k-1}$ &$f_{1,k}=\sum\limits_{j=0}^{k-1}f_{1,j}$ & $f_{1,k+1}=\sum\limits_{(i,j)\in Z_1}\beta_{i,j} f_{i,j}$  \\
\hline $\vdots$  &$\ddots$ & $\vdots$ &$\vdots$&$\vdots$  \\
\hline $f_{N-1,0}$ &$\cdots$ & $f_{N-1,k-1}$ &$f_{N-1,k}=\sum\limits_{j=0}^{k-1}f_{N-1,j}$ & $f_{N-1,k+1}=\sum\limits_{(i,j)\in Z_{N-1}}\beta_{i,j} f_{i,j}$  \\
\hline
\end{tabular}
\end{center}
\end{table*}

Given an integer $k\ge2$, recursively define $k$  matrices  $A^{(k)}_0,\cdots,A^{(k)}_{k-1}$ of order $N$ over $\mathbf{F}_3$ as
\begin{equation}\label{coding matrix}
  A^{(k)}_0=I_{2^{k-1}},\ \ A^{(k)}_1=\left(
                                    \begin{array}{cc}
                                       & -I_{2^{k-2}} \\
                                      I_{2^{k-2}} &  \\
                                    \end{array}
                                  \right), \ \ A^{(k)}_j=\left(
                                \begin{array}{cc}
                                  A^{(k-1)}_{j-1} &  \\
                                   & -A^{(k-1)}_{j-1} \\
                                \end{array}
                              \right)\mbox{\ for \ } 2\le j< k
\end{equation}
where
\begin{equation*}
  A^{(2)}_0=I_{2},\ \ A^{(2)}_1=\left(
                                \begin{array}{cc}
                                  0 & -1 \\
                                  1 & 0 \\
                                \end{array}
                              \right).
\end{equation*}

First of all, the following properties of the matrices in \eqref{coding matrix} are obvious.

\begin{Property}\label{Prop_Matrix}
For any $k\ge 2$, the matrix $A^{(k)}_{j}$ in  \eqref{coding matrix} with  $1\le j< k$ satisfies

(i) $(A^{(k)}_{j})^2=-I_{2^{k-1}}$;

(ii)  Both each row and each column of $A_j^{(k)}$ have only one nonzero entry.

\end{Property}

Next, we show that   the matrix $A^{(k)}_{j}$ in  \eqref{coding matrix} is just the
coding matrix for systematic node $j$ of the $(k+2,k)$ Zigzag code for all $0\le j< k$.

\begin{Theorem}\label{second parity}
The coding matrices of  the $(k+2,k)$ Zigzag code are $A_0^{(k)},\cdots,A_{k-1}^{(k)}$, i.e.,
\begin{equation*}
  \mathbf{f}_{k+1}=A_0^{(k)}\mathbf{f}_0+\cdots + A_{k-1}^{(k)}\mathbf{f}_{k-1}
\end{equation*}
where $\mathbf{f}_j=(f_{0,j},\cdots,f_{N-1,j})^T$.
\end{Theorem}

\begin{proof}
Let $A(l,i)$ denote the entry at row $l$  and column $i$ of matrix $A$. By Property 1-(ii), equations \eqref{Eqn P and P-1} and \eqref{beta}, it suffices to prove $A_j^{(k)}(l,P_j^{-1}(l))=\beta_{P_j^{-1}(l),j}$, i.e.,
\begin{equation}\label{Eqn A and e0}
  A_0^{(k)}(l,l)=A_0^{(k)}(l,l\oplus e_0)=\beta_{l,0}=1, 0\le l< N
\end{equation}
and
\begin{equation}\label{Eqn A and e}
  A_j^{(k)}(l,l\oplus e_j)=\beta_{l\oplus e_j,j}=(-1)^{l_1+\cdots+l_j+1},\ \  1\le j< k, \\ 0\le l< N.
\end{equation}

Obviously, \eqref{Eqn A and e0} holds since $A_0^{(k)}$ is the identity matrix  and \eqref{Eqn A and e} holds for $j=1$, i.e.,
$A_1^{(k)}(l,l\oplus e_1)=(-1)^{l_1+1}, 0\le l<N $, by the definition in \eqref{coding matrix}.

Hereafter, we prove \eqref{Eqn A and e} for $j\ge 2$ by the induction.
Suppose that \eqref{Eqn A and e} holds for $k\ge 2$ and $1\le j< k$. Then,
\begin{eqnarray*}
 &&A_j^{(k+1)}(l,l\oplus e_j)\\ &=& A_j^{(k+1)}((l_1,\cdots,l_{k}),(l_1,\cdots,l_{j-1},l_{j}\oplus 1,l_{j+1},\cdots,l_{k}))\\
   &=& (-1)^{l_1}A_{j-1}^{(k)}((l_2,\cdots,l_{k}),(l_2,\cdots,l_{j-1},l_{j}\oplus 1,l_{j+1},\cdots,l_{k}))  \\
   &=&(-1)^{l_1+\cdots+l_{j}+1}
\end{eqnarray*}
for $2\le j< k+1$ and $0\le l<2^{k}$, where the last two equalities respectively follow from  \eqref{coding matrix} and the assumption.

\end{proof}

\section{Bounds on disk I/O to optimally repair the parity nodes of the Zigzag code}

For a general $(k+2,k)$ MSR code over $\mathbf{F}_q$ defined in Table I,  Wang \emph{et al.} \cite{MDR} proved that the minimal disk I/O to repair the first and  second parity nodes are respectively at least $(k+1)N/2$ and $kN$ if $q=2$. In fact, the assertion can be proved for $q>2$ by almost the same proof in \cite{MDR}.

Specifically for the Zigzag code,
in this section we give a more tight bound on the minimal disk I/O for the  optimal repair of the  parity nodes.

Firstly, we state a connection between the optimal repair strategies for the two parity nodes of the Zigzag code.

\begin{Lemma}\label{repair relation parity 1 and 2}
If $S^{(k)}$ and $\tilde{S}^{(k)}$ are the repair matrices for the first parity node of the $(k+2,k)$ Zigzag code, then $\tilde{S}^{(k)} A_j^{(k)},0\le j< k$, and $S^{(k)}$  are the repair matrices for the second parity node, and vice versa.
\end{Lemma}

\begin{proof} Note from \eqref{Eqn_MDS} and \eqref{coding matrix} that $A_{0}^{(k)}-A_{l}^{(k)}=I_{N}-A_{l}^{(k)}$ is nonsingular
for $1\le l< k$. Then,
\begin{eqnarray}\label{Eqn_Connection}
\textrm{rank} \left(\left(
\begin{array}{c}
\tilde{S}^{(k)} \\
S^{(k)} \left((A_0^{(k)})^{-1}-(A_l^{(k)})^{-1}\right)
\end{array}
\right)\right) &=&\textrm{rank}\left(\left(
\begin{array}{c}
\tilde{S}^{(k)} \\
S^{(k)} (I_{N}+A_l^{(k)})
\end{array}
\right)\right)\nonumber\\
&=&\textrm{rank}\left(\left(
\begin{array}{c}
\tilde{S}^{(k)} \\
S^{(k)} (I_{N}+A_l^{(k)})
\end{array}
\right)(I_{N}-A_l^{(k)})\right)\nonumber\\
&=&\textrm{rank}\left(\left(
\begin{array}{c}
\tilde{S}^{(k)}(I_{N}-A_l^{(k)}) \\
S^{(k)} (I_{N}+A_l^{(k)})(I_{N}-A_l^{(k)})
\end{array}
\right)\right)\nonumber\\
&=&\textrm{rank}\left(\left(
\begin{array}{c}
S^{(k)}\\
\tilde{S}^{(k)}(I_{N}-A_l^{(k)})
\end{array}
\right)\right)\nonumber\\
&=&\textrm{rank}\left(\left(
\begin{array}{c}
S^{(k)}\\
\tilde{S}^{(k)}(A_0^{(k)}-A_l^{(k)})
\end{array}
\right)\right)
\end{eqnarray}
where in the first and fourth identities we use Property 1-(i), i.e., $(A_{l}^{(k)})^2=-I_{N}$ and then $(A_{l}^{(k)})^{-1}=-A_{l}^{(k)}$.

In addition,
\begin{eqnarray*}
\textrm{rank} \left(\left(
\begin{array}{c}
\tilde{S}^{(k)} \\
S^{(k)}(A_0^{(k)})^{-1}
\end{array}
\right)\right) = \textrm{rank} \left(\left(
\begin{array}{c}
\tilde{S}^{(k)} \\
S^{(k)}
\end{array}
\right)\right) =\textrm{rank} \left(\left(
\begin{array}{c}
S^{(k)} \\
\tilde{S}^{(k)}A_0^{(k)}
\end{array}
\right)\right).
\end{eqnarray*}

Therefore, the result can be obtained from \eqref{repair_1parity_node_requirement1}, \eqref{repair_1parity_node_requirement2}, \eqref{repair_2parity_node_requirement1} and \eqref{repair_2parity_node_requirement2}.
\end{proof}


\begin{Theorem}\label{I/O bound}
The disk I/O to optimally repair the first or second parity node  of the $(k+2,k)$ Zigzag code is at least $kN+{k-3\over 2(k-1)}N$.
\end{Theorem}

\begin{proof} Suppose that $S_a^{(k)}$ and $\tilde{S}_a^{(k)}$  are two repair matrices for the first parity node of $(k+2,k)$ Zigzag code.
According to the definition of repair disk I/O,  we need to prove $kN_1+N_2\ge kN+{k-3\over 2(k-1)}N$, where $N_1$ and $N_2$ respectively denote the number of nonzero columns of
the matrices $S_a^{(k)}$ and $\tilde{S}_a^{(k)}$.

By \eqref{repair_1parity_node_requirement1} and \eqref{repair_1parity_node_requirement2}, we have
\begin{eqnarray}\label{repair_1parity_node_requirement1 zigzag}
\textrm{rank} \left(\left(
\begin{array}{c}
S_a^{(k)} \\
\tilde{S}_a^{(k)}A_0^{(k)}
\end{array}
\right)\right) =\textrm{rank} \left(\left(
\begin{array}{c}
S_a^{(k)} \\
\tilde{S}_a^{(k)}
\end{array}
\right)\right) =N
\end{eqnarray}
and
\begin{eqnarray}\label{repair_1parity_node_requirement2 zigzag}
\textrm{rank} \left(\left(
\begin{array}{c}
S_a^{(k)} \\
\tilde{S}_a^{(k)} (A_0^{(k)}-A_l^{(k)})
\end{array}
\right)\right)  =\textrm{rank} \left(\left(
\begin{array}{c}
S_a^{(k)} \\
\tilde{S}_a^{(k)} (I_N-A_l^{(k)})
\end{array}
\right)\right)  = {N\over 2},\ \ 1\le l< k.
\end{eqnarray}

For $0\le i<N$, denote by $S_a^{(k)}[i]$ and $\tilde{S}_a^{(k)}[i]$ the column $i$ of $S_a^{(k)}$ and $\tilde{S}_a^{(k)}$.
Assume that columns $i_1,i_2,\cdots,i_{N-N_1}$ of $S_a^{(k)}$ are zero columns.
Note that in \eqref{repair_1parity_node_requirement2 zigzag}, $\mathrm{rank}(S_a^{(k)})=\mathrm{rank}(\tilde{S}_a^{(k)} (I_N-A_l^{(k)}))=N/2$.
Then, we have that $\tilde{S}_a^{(k)}(I_N-A_l^{(k)})[i_s]=\tilde{S}_a^{(k)}[i_s]-(\tilde{S}_a^{(k)}A_l^{(k)})[i_s]$
is also a zero column, i.e.,
\begin{equation*}
(\tilde{S}_a^{(k)}A_l^{(k)})[i_s]=\tilde{S}_a^{(k)}[i_s] ~\mathrm{for}~ 1\le l< k ~\mathrm{and}~1\le s\le N-N_1.
\end{equation*}
Further, it follows from Property 1-(ii) and \eqref{Eqn A and e} that only the $(i\oplus e_l)$th entry in  $A_l^{(k)}[i]$ is  $\pm1$, which implies   $(\tilde{S}_a^{(k)}A_l^{(k)})[i_s]=\pm\tilde{S}_a^{(k)}[i_s\oplus e_l]$. Thus,
\begin{equation}\label{i columns equal}
 \tilde{S}_a^{(k)}[i_s\oplus e_l]=\pm \tilde{S}_a^{(k)}[i_s]\ \ \mbox{for\ \ } 1\le l< k ~\mathrm{and}~1\le s\le N-N_1.
\end{equation}

On the other hand, it is seen from \eqref{repair_1parity_node_requirement1 zigzag} that all the columns $i_1,i_2,\cdots,i_{N-N_1}$ of $\tilde{S}_a^{(k)}$ are
linearly independent, which indicates that \begin{equation}\label{columns disjoint}
  \{i_u\oplus e_l:1\le l< k\}\cap \{i_v\oplus e_l:1\le l< k\}=\emptyset \mbox{\ \ for\ \ }1\le u\ne v\le N-N_1.
\end{equation}

Therefore, applying  \eqref{i columns equal} and \eqref{columns disjoint} to $\textrm{rank}(\tilde{S}_a^{(k)})=N/2$,
we obtain $N/2\le N-(k-1)(N-N_1)$, i.e., $N_1\ge N-{N\over {2(k-1)}}$.
By means of \eqref{Eqn_Connection}, we can prove $N_2\ge N-{N\over {2(k-1)}}$ in the same fashion. Hence,
\begin{equation*}
  kN_1+N_2\ge (k+1)(N-{N\over {2(k-1)}})=kN+N-{N(k+1)\over {2(k-1)}}=kN+{k-3\over 2(k-1)}N.
\end{equation*}
That is, the assertion is valid for the first parity node.

For the second parity node of the $(k+2,k)$ Zigzag code, assume that $S_b^{(k)}A_j^{(k)},0\le j< k$, and $\tilde{S}_b^{(k)}$  are the repair matrices.
According to the definition, the repair disk I/O is the total number of  nonzero columns of the matrices
$S_b^{(k)}A_j^{(k)}$ and $\tilde{S}_b^{(k)}, 0\le j< k$, which is $kN_1+N_2$ by Property \ref{Prop_Matrix}-(ii), where $N_1$ and $N_2$ respectively denote the number of nonzero columns of the matrices $S_b^{(k)}$ and $\tilde{S}_b^{(k)}$. By Lemma \ref{repair relation parity 1 and 2}, it is known that
$\tilde{S}_b^{(k)}$ and $S_b^{(k)}$ are two repair matrices for the first parity node. Therefore, by the analysis for the
first parity node we have
$N_1\ge N-{N\over {2(k-1)}}$ and $N_2\ge N-{N\over {2(k-1)}}$, i.e., $kN_1+N_2\ge kN+{k-3\over 2(k-1)}N$.
\end{proof}

\section{Repair matrices for the parity nodes of the Zigzag code}

In this section, we give the repair matrices for the parity nodes of the $(k+2,k)$ Zigzag code and verify that they satisfy \eqref{repair_1parity_node_requirement1}, \eqref{repair_1parity_node_requirement2}, \eqref{repair_2parity_node_requirement1} and \eqref{repair_2parity_node_requirement2}.

Recursively define the $2^{k-2}\times2^{k-1}$  matrices $E^{(k)}$ and $F^{(k)}$  over  $\mathbf{F}_3$  as
\begin{equation}\label{Eqn EF}
  E^{(k)}=\left(
              \begin{array}{cc}
                E^{(k-1)} &  \\
                 & F^{(k-1)} \\
              \end{array}
            \right),
  \ \  F^{(k)}=\left(
              \begin{array}{cc}
                F^{(k-1)} &  \\
                 & E^{(k-1)} \\
              \end{array}
            \right), \ \ k\ge3
\end{equation}
where
\begin{equation}\label{EF initial 1}
  E^{(2)}=\left(
               \begin{array}{cc}
                 0 & -1\\
               \end{array}
             \right),\ \  F^{(2)}=\left(
               \begin{array}{cc}
                 -1 & 0 \\
               \end{array}
             \right).\end{equation}
Next recursively define the  $2^{k-2}\times2^{k-1}$  matrices $S^{(k)}_a$ and $\tilde{S}^{(k)}_a$  over  $\mathbf{F}_3$  as
\begin{equation}\label{Eqn S}
  S^{(k)}_a=\left(
              \begin{array}{cc}
                S^{(k-1)}_a & E^{(k-1)} \\
                 & \tilde{S}^{(k-1)}_a \\
              \end{array}
            \right)
  ,\ \ \tilde{S}^{(k)}_a=\left(
                           \begin{array}{cc}
                             \tilde{S}^{(k-1)}_a & -F^{(k-1)}\\
                              &  S^{(k-1)}_a \\
                           \end{array}
                         \right), \ \ k\ge3
  \end{equation}
where  \begin{equation}\label{S initial 1} \ \  S^{(2)}_a=\left(
               \begin{array}{cc}
                 0 & 1 \\
               \end{array}
             \right),\ \ \tilde{S}^{(2)}_a=\left(
               \begin{array}{cc}
                 1 & 1 \\
               \end{array}
             \right).
\end{equation}

\begin{Proposition}\label{full rank}
 For   $k\ge2$,
$\mathrm{rank}\left(\left(
                 \begin{array}{c}
                   S_a^{(k)} \\
                   \tilde{S}_a^{(k)}A_0^{(k)}
                 \end{array}
               \right)\right)=N$.
\end{Proposition}

\begin{proof}
When $k=2$, the statement is easily checked.
For any given $k\ge 2$, suppose that the statement is true.
According to recursive definition in \eqref{Eqn S},  we have
\begin{eqnarray*}
   \mathrm{rank}\left(\left(
                 \begin{array}{c}
                   S_a^{(k+1)} \\
                   \tilde{S}_a^{(k+1)}A_0^{(k+1)}
                 \end{array}
               \right)\right)&=& \mathrm{rank}\left(\left(
                 \begin{array}{c}
                   S_a^{(k+1)} \\
                   \tilde{S}_a^{(k+1)}
                 \end{array}
               \right)\right)\\&=&\mathrm{rank}\left(\left(
                 \begin{array}{cc}
                   S_a^{(k)} & E^{(k)}\\
                   & \tilde{S}_a^{(k)} \\
                   \tilde{S}_a^{(k)} & -F^{(k)}\\
                   & S_a^{(k)}
                 \end{array}
               \right)\right)\\&=& \mathrm{rank}\left(
\left(
                 \begin{array}{cc}
                   S_a^{(k)} & E^{(k)}\\
                   \tilde{S}_a^{(k)} & -F^{(k)}\\
                    & \tilde{S}_a^{(k)} \\
                   & S_a^{(k)}
                 \end{array}
               \right)\right)\\
               &=&2N
\end{eqnarray*}
since $\left(
                 \begin{array}{c}
                   S_a^{(k)} \\
                   \tilde{S}_a^{(k)}
                 \end{array}
               \right)=\left(
                 \begin{array}{c}
                   S_a^{(k)} \\
                   \tilde{S}_a^{(k)}A_0^{(k)}
                 \end{array}
               \right)$ is an $N\times N$ matrix of full rank.

Thus, the proof is finished by the above induction.
\end{proof}

\begin{Proposition}\label{rank A1-A0}
For $k\ge 2$,
$\mathrm{rank}\left(\left(
                 \begin{array}{c}
                   S_a^{(k)} \\
                   \tilde{S}_a^{(k)}(A_0^{(k)}-A_{1}^{(k)}) \\
                 \end{array}
               \right)\right)=N/2$.
\end{Proposition}

\begin{proof}
When $k=2$, the statement is easily checked. When $k>2$,
by the recursive definitions in \eqref{coding matrix} and \eqref{Eqn S}, we have
\begin{eqnarray*}
&&\tilde{S}_a^{(k)}(A_0^{(k)}-A_{1}^{(k)})\\&=& \tilde{S}_a^{(k)}(I_N-A_{1}^{(k)})\\
 &=&\left(
                                               \begin{array}{cc}
                                                 \tilde{S}_a^{(k-1)} & -F^{(k-1)} \\
                                                  & S_a^{(k-1)} \\
                                               \end{array}
                                             \right)\left(
                                             \left(\begin{array}{cc}
                                                       I_{N/2} &\\
                                                       &  I_{N/2}
                                                       \end{array}
                                                     \right)-\left(
                                                       \begin{array}{cc}
                                                           & -I_{N/2}\\
                                                         I_{N/2} &
                                                       \end{array}
                                                     \right)\right)\\&=&\left(
                                                                 \begin{array}{cc}
                                                                  \tilde{S}_a^{(k-1)}+F^{(k-1)} & \tilde{S}_a^{(k-1)}-F^{(k-1)} \\
                                                                   -S_a^{(k-1)} & S_a^{(k-1)} \\
                                                                 \end{array}
                                                               \right).
\end{eqnarray*}
Therefore,
\begin{eqnarray}\label{Eqn_Matrx_A1}
&&  \mathrm{rank}\left(\left(
                 \begin{array}{c}
                   S_a^{(k)} \\
                  \tilde{S}_a^{(k)}(A_0^{(k)}-A_{1}^{(k)})
                 \end{array}
               \right)\right)\nonumber\\ &=& \mathrm{rank}\left(\left(
                                                                 \begin{array}{cc}
                                                                 S_a^{(k-1)} & E^{(k-1)}\\
                                                                  & \tilde{S}_a^{(k-1)}\\
                                                                  \tilde{S}_a^{(k-1)}+F^{(k-1)} & \tilde{S}_a^{(k-1)}-F^{(k-1)} \\
                                                                   -S_a^{(k-1)} & S_a^{(k-1)} \\
                                                                 \end{array}
                                                               \right)\right)\nonumber\\
   &=& \mathrm{rank}\left( P\cdot
   \left(
                                                                 \begin{array}{cc}
                                                                 S_a^{(k-1)} & E^{(k-1)}\\
                                                                  & \tilde{S}_a^{(k-1)}\\
                                                                 \tilde{S}_a^{(k-1)}+F^{(k-1)} & \tilde{S}_a^{(k-1)}-F^{(k-1)} \\
                                                                   -S_a^{(k-1)} & S_a^{(k-1)} \\
                                                                 \end{array}
                                                               \right)\cdot Q\right)\nonumber\\
                                                                &=& \mathrm{rank}\left(\left(
                                                                 \begin{array}{cc}
                                                                 S_a^{(k-1)}+E^{(k-1)} &  \\
                                                                 \tilde{S}_a^{(k-1)} &  \\
                                                                   &\tilde{S}_a^{(k-1)}+F^{(k-1)} \\
                                                                    & S_a^{(k-1)} \\
                                                                 \end{array}
                                                               \right)\right)\nonumber\\
                                                                &=&\mathrm{rank}\left(\left(
                                                                 \begin{array}{c}
                                                                 S_a^{(k-1)}+E^{(k-1)} \\
                                                                 \tilde{S}_a^{(k-1)} \\
                                                                 \end{array}
                                                               \right)\right)+\mathrm{rank}\left(\left(
                                                                 \begin{array}{c}
                                                                  \tilde{S}_a^{(k-1)}+F^{(k-1)} \\
                                                                   S_a^{(k-1)} \\
                                                                 \end{array}
                                                               \right)\right)
                                                               \end{eqnarray}
where the two matrices $P,Q$ are respectively defined by
\begin{eqnarray*}
P=\left(\begin{array}{cccc}
                            I_{N/4} &  &  & I_{N/4} \\
                             & I_{N/4} &  &  \\
                             & -I_{N/4} & -I_{N/4} &  \\
                             &  &  & I_{N/4}\end{array}
                        \right),~Q=\left( \begin{array}{cc}
       I_{N/2} & -I_{N/2} \\
        I_{N/2} &  \\
     \end{array}
   \right).
\end{eqnarray*}

Next, we prove
\begin{equation*}
  \mathrm{rank}\left(\left(
                                                                 \begin{array}{c}
                                                                 S_a^{(k)}+E^{(k)} \\
                                                                 \tilde{S}_a^{(k)} \\
                                                                 \end{array}
                                                               \right)\right)=\mathrm{rank}\left(\left(
                                                                 \begin{array}{c}
                                                                  \tilde{S}_a^{(k)}+F^{(k)} \\
                                                                   S_a^{(k)} \\
                                                                 \end{array}
                                                               \right)\right)=N/2
\end{equation*}
for any $k\ge 2$ by the induction.

When $k=2$, the statement is easily verified. For any $k \ge 2$, suppose that
it is true. By the definition of $S^{(k+1)}_a$ and $\tilde{S}^{(k+1)}_a$  in \eqref{Eqn S}, we then have

\begin{eqnarray*}
  && \mathrm{rank}\left(\left(
                                                                 \begin{array}{c}
                                                                 S_a^{(k+1)}+E^{(k+1)} \\
                                                                 \tilde{S}_a^{(k+1)} \\
                                                                 \end{array}
                                                               \right)\right)\\
   &=&\mathrm{rank}\left(\left(
                                                                 \begin{array}{cc}
                                                                 S_a^{(k)}+E^{(k)} & E^{(k)}\\
                                                                 & \tilde{S}_a^{(k)}+F^{(k)}\\
                                                                 \tilde{S}_a^{(k)} &-F^{(k)}\\
                                                                 & S_a^{(k)}\\
                                                                 \end{array}
                                                               \right)\right)\\ &=& \mathrm{rank}\left(\left(
                                                                                                       \begin{array}{cccc}
                                                                                                         I_{N/2}&  &  & I_{N/2} \\
                                                                                                           & I_{N/2} & I_{N/2} &   \\
                                                                                                           & I_{N/2} &   &   \\
                                                                                                           &   &   & I_{N/2} \\
                                                                                                       \end{array}
                                                                                                     \right)
                                                               \left(
                                                                 \begin{array}{cc}
                                                                 S_a^{(k)}+E^{(k)} & E^{(k)}\\
                                                                 & \tilde{S}_a^{(k)}+F^{(k)}\\
                                                                 \tilde{S}_a^{(k)} &-F^{(k)}\\
                                                                 & S_a^{(k)}\\
                                                                 \end{array}
                                                               \right)\left(
                                                                        \begin{array}{cc}
                                                                          I_{N} & -I_{N} \\
                                                                           & I_{N} \\
                                                                        \end{array}
                                                                      \right)
                                                               \right) \\
   &=&\mathrm{rank}\left(\left(
                                                                 \begin{array}{cc}
                                                                 S_a^{(k)}+E^{(k)} & \\
                                                                 \tilde{S}_a^{(k)}&\\
                                                                 &\tilde{S}_a^{(k)} +F^{(k)}\\
                                                                 & S_a^{(k)}\\
                                                                 \end{array}
                                                               \right)\right)  \\
   &=&   \mathrm{rank}\left(\left(
                                                                 \begin{array}{c}
                                                                 S_a^{(k)}+E^{(k)} \\
                                                                 \tilde{S}_a^{(k)} \\
                                                                 \end{array}
                                                               \right)\right)+\mathrm{rank}\left(\left(
                                                                 \begin{array}{c}
                                                                  \tilde{S}_a^{(k)}+F^{(k)} \\
                                                                   S_a^{(k)} \\
                                                                 \end{array}
                                                               \right)\right)\\
                                                               &=&N
\end{eqnarray*}
where the last identity comes from the assumption. Similarly, we can get $\mathrm{rank}\left(\left(
                                                                 \begin{array}{c}
                                                                  \tilde{S}_a^{(k+1)}+F^{(k+1)} \\
                                                                   S_a^{(k+1)} \\
                                                                 \end{array}
                                                               \right)\right)=N$.
This completes the proof after substituted into  \eqref{Eqn_Matrx_A1}.

\end{proof}

\begin{Proposition}\label{rank Ai-A0} Given $k\ge 3$,
$\mathrm{rank}\left(\left(
                 \begin{array}{c}
                   S_a^{(k)} \\
                   \tilde{S}_a^{(k)}(A_{0}^{(k)}-A_{i}^{(k)}) \\
                 \end{array}
               \right)\right)=N/2$ for all $2\le i< k$.
\end{Proposition}

\begin{proof} If $k=3$, the statement is obvious. For any $k\ge3$, assume that
it is true for all $2\le j< k$.
When $j\ge 2$, according to  the  definitions of $A_j^{(k+1)}$ in \eqref{coding matrix} and $S^{(k+1)}_a,\tilde{S}^{(k+1)}_a$  in \eqref{Eqn S},
\begin{eqnarray*}
\mathrm{rank}\left(\left(
                 \begin{array}{c}
                   S_a^{(k+1)} \\
                   \tilde{S}_a^{(k+1)}(A_{0}^{(k+1)}-A_{j}^{(k+1)}) \\
                 \end{array}
               \right)\right)
   = \mathrm{rank}\left(\left(
                 \begin{array}{c}
                   S_a^{(k+1)} \\
                   \tilde{S}_a^{(k+1)}(I_{2N}-A_{j}^{(k+1)}) \\
                 \end{array}
               \right)\right)
   = \mathrm{rank}\left(\left(
                                                       \begin{array}{cc}
                                                         U_j^{(k)} & W_j^{(k)} \\
                                                           & V_j^{(k)} \\
                                                       \end{array}
                                                     \right)
               \right)
  \end{eqnarray*}
for three $N\times N$ matrices
\begin{eqnarray*}
U_{j}^{(k)}=\left(
                 \begin{array}{c}
                   S_a^{(k)} \\
                   \tilde{S}_a^{(k)}(I_{N}-A_{j-1}^{(k)}) \\
                 \end{array}
               \right),V_{j}^{(k)}=\left(
                 \begin{array}{c}
                 \tilde{S}_a^{(k)} \\
             S_a^{(k)}(I_{N}+A_{j-1}^{(k)}) \\
                 \end{array}\right),W_{j}^{(k)}=\left(
                 \begin{array}{c}
                E^{(k)} \\
                     -F^{(k)}(I_{N}+A_{j-1}^{(k)}) \\
                 \end{array}
               \right),
\end{eqnarray*}
by the recursive definitions which satisfy
 \begin{eqnarray*}
 W_j^{(k)}=\left\{
             \begin{array}{cl}
             -U_j^{(k)} +R^{(k)}V_j^{(k)}, &\mbox{\ if\ }j=2\\
               U_j^{(k)}Q^{(k)} -P^{(k)}V_j^{(k)}, &\mbox{\ if\ }j>2
             \end{array}
           \right.
 \end{eqnarray*}
where
\begin{eqnarray*}
R^{(k)}=\left(
                          \begin{array}{cccc}
                          \textbf{0}_{N/4}&I_{N/4}&I_{N/4}&\textbf{0}_{N/4}\\
                            -I_{N/4} &  \textbf{0}_{N/4}&\textbf{0}_{N/4}  & I_{N/4} \\
                           \textbf{0}_{N/4} &\textbf{0}_{N/4}&\textbf{0}_{N/4}&I_{N/4} \\
                           \textbf{0}_{N/4}&\textbf{0}_{N/4}&-I_{N/4}&\textbf{0}_{N/4} \\
                          \end{array}
                        \right),\ \ P^{(k)}=\left(
                          \begin{array}{cccc}
                          \textbf{0}_{N/4}&\textbf{0}_{N/4}&\textbf{0}_{N/4}&\textbf{0}_{N/4}\\
                            I_{N/4} &  \textbf{0}_{N/4}&\textbf{0}_{N/4}  & \textbf{0}_{N/4} \\
                           \textbf{0}_{N/4} &\textbf{0}_{N/4}&\textbf{0}_{N/4}&\textbf{0}_{N/4} \\
                           \textbf{0}_{N/4}&\textbf{0}_{N/4}&I_{N/4}&\textbf{0}_{N/4} \\
                          \end{array}
                        \right),\ \ Q^{(k)}=\left(
            \begin{array}{cc}
             \textbf{0}_{N/2} & \textbf{0}_{N/2} \\
               I_{N/2}& \textbf{0}_{N/2} \\
            \end{array}
          \right)
\end{eqnarray*}
and $\textbf{0}_\textsf{N}$ denotes the zero matrix of order $\textsf{N}$.

Hence,
\begin{eqnarray}\label{Eqn_Matrix_11}
&&\mathrm{rank}\left(\left(
                 \begin{array}{c}
                   S_a^{(k+1)} \\
                   \tilde{S}_a^{(k+1)}(A_{0}^{(k+1)}-A_{j}^{(k+1)}) \\
                 \end{array}
               \right)\right)
   \nonumber\\ &=&  \mathrm{rank}\left(\left(
                 \begin{array}{cc}
                   S_a^{(k)}\\
                   \tilde{S}_a^{(k)}(I_{N}-A_{j-1}^{(k)}) \\
                 \end{array}
               \right)\right)+\mathrm{rank}\left(\left(
                 \begin{array}{cc}
                   \tilde{S}_a^{(k)} \\
             S_a^{(k)}(I_{N}+A_{j-1}^{(k)}) \\
                 \end{array}
               \right)\right)
\end{eqnarray}
for $j\ge 2$.

Further, note from \eqref{Eqn_MDS} that $A_{0}^{(k)}-A_{j-1}^{(k)}=I_{N}-A_{j-1}^{(k)}$ is nonsingular if $j\ge 2$.
Then,
\begin{eqnarray*}\label{Eqn_Matrix_13}
&& \mathrm{rank}\left(\left(
                 \begin{array}{cc}
                   \tilde{S}_a^{(k)} \\
             S_a^{(k)}(I_{N}+A_{j-1}^{(k)}) \nonumber\\
                 \end{array}
               \right)\right)\\ \nonumber
   &=&  \mathrm{rank}\left(\left(
                 \begin{array}{cc}
                   S_a^{(k)}\\
                   \tilde{S}_a^{(k)}(I_{N}-A_{j-1}^{(k)}) \\
                 \end{array}
               \right)\right)\\
   &=& N/2
\end{eqnarray*}
where in the first identity we use \eqref{Eqn_Connection} and the in last identity we use
the assumption if $j\ge 3$ and Proposition \ref{rank A1-A0} if $j=2$. This completes the proof after substituted into  \eqref{Eqn_Matrix_11}.

\end{proof}

The following main result  is immediate.

\begin{Theorem}\label{Thm parity 1}
$S^{(k)}_a$ and $\tilde{S}^{(k)}_a$ that defined by \eqref{Eqn EF}, \eqref{EF initial 1}, \eqref{Eqn S} and \eqref{S initial 1} are  the  repair matrices for   the first parity node of the $(k+2,k)$ Zigzag code, whose repair disk I/O  is $kN+N-k$.
\end{Theorem}
\begin{proof} The optimal repair property of repair matrices $S^{(k)}_a$ and $\tilde{S}^{(k)}_a$ is obvious from Propositions \ref{full rank}, \ref{rank A1-A0} and \ref{rank Ai-A0}.

 Note that there is only one zero column in $S_a^{(k)}$ and no zero columns in $\tilde{S}_a^{(k)}$, which means $N-1$ elements should be read in each of the systematic nodes and all the $N$ elements should be read in  the second parity node to repair  the first parity node. Thus the disk I/O to repair the first parity node is $kN+N-k$.
\end{proof}

By Lemma \ref{repair relation parity 1 and 2},  the second parity node of the $(k+2,k)$ Zigzag code can also be optimally repaired. However, if we use $S_b^{(k)}A_i^{(k)},0\le i< k$ and $\tilde{S}_b^{(k)}$ as the repair matrices, where $S_b^{(k)}=\tilde{S}^{(k)}_a$ and $\tilde{S}_b^{(k)}=S^{(k)}_a$ are defined by \eqref{Eqn EF}, \eqref{EF initial 1}, \eqref{Eqn S} and \eqref{S initial 1}, then its repair disk
I/O will be $kN+N-1$  since $S_a^{(k)}$ has only one  zero column and $\tilde{S}_a^{(k)}A_i^{(k)}$ has no zero columns for $0\le i< k$. In the following, by choosing another initial values of $E^{(2)}$, $F^{(2)}$, $S_a^{(2)}$, $\tilde{S}_a^{(2)}$ in \eqref{EF initial 1} and \eqref{S initial 1}, the   disk I/O to optimally repair   the second parity node can also be reduced to $kN+N-k$.

Reset
\begin{equation}\label{EFS initial 2}
  E^{(2)}=\left(
               \begin{array}{cc}
                 -1 & 0\\
               \end{array}
             \right),\ \  F^{(2)}=\left(
               \begin{array}{cc}
                 0 & -1 \\
               \end{array}
             \right), \ \  S^{(2)}_a=\left(
               \begin{array}{cc}
                 1 & -1 \\
               \end{array}
             \right),\ \ \tilde{S}^{(2)}_a=\left(
               \begin{array}{cc}
                 0 & 1 \\
               \end{array}
             \right),
\end{equation}
then we have the following result.

\begin{Theorem}\label{Thm parity 2}
Let $S^{(k)}_a$ and $\tilde{S}^{(k)}_a$ be defined by \eqref{EFS initial 2}, \eqref{Eqn EF} and \eqref{Eqn S}, then $S_b^{(k)}A_i^{(k)},0\le i< k$ and $\tilde{S}_b^{(k)}$  are the repair matrices for  the second parity node of the $(k+2,k)$ Zigzag code where $S_b^{(k)}=\tilde{S}^{(k)}_a$ and $\tilde{S}_b^{(k)}=S^{(k)}_a$. Moreover,
the disk I/O to optimally repair the second parity node is $kN+N-k$.
\end{Theorem}
\begin{proof}
Firstly, it can be easily verified that the results in Propositions \ref{full rank}, \ref{rank A1-A0} and \ref{rank Ai-A0} are also hold for $S^{(k)}_a$ and $\tilde{S}^{(k)}_a$ defined from the initial values $E^{(2)}$, $F^{(2)}$, $S_a^{(2)}$ and $\tilde{S}_a^{(2)}$ in \eqref{EFS initial 2}. Secondly, it follows from Lemma \ref{repair relation parity 1 and 2} that $\tilde{S}_a^{(k)}A_i^{(k)},0\le i< k$ and $S_a^{(k)}$  are the  repair matrices for the second parity node of the $(k+2,k)$ Zigzag code.
\end{proof}

From Theorems \ref{Thm parity 1} and \ref{Thm parity 2}, it is seen that the disk I/O to optimally repair the parity nodes of the Zigzag code is very close to the lower bound given in Lemma \ref{I/O bound}.

Finally, we give some   examples of the repair matrices for the  parity nodes of the $(k+2,k)$ Zigzag code.
\begin{Example}
The first parity node of the $(5,3)$ Zigzag code,  $(6,4)$ Zigzag code, and $(7,5)$ Zigzag code, can be respectively optimally
repaired by the following matrices
\begin{equation*}
  S_a^{(3)}=\left(
        \begin{array}{cccc}
                    0 & 1 & 0 & -1 \\
                    0 & 0 & 1 & 1 \\
        \end{array}
      \right),\ \ \tilde{S}_a^{(3)}=\left(
        \begin{array}{cccc}
          1 & 1 & 1 & 0 \\
          0 & 0 & 0 & 1 \\
        \end{array}
      \right)
\end{equation*}

\begin{equation*}
  S_a^{(4)}=\left(
        \begin{array}{cccccccc}
          0 & 1 & 0 & -1 & 0 & -1 & 0 & 0 \\
          0 & 0 & 1 & 1 & 0 & 0 & -1 & 0 \\
          0 & 0 & 0 & 0 & 1 & 1 & 1 & 0 \\
          0 & 0 & 0& 0 & 0 & 0 & 0 & 1 \\
        \end{array}
      \right)
  ,\ \ \tilde{S}_a^{(4)}=\left(
        \begin{array}{cccccccc}
          1 & 1 & 1 & 0 & 1 & 0 & 0 & 0 \\
          0 & 0 & 0 & 1 & 0 & 0 & 0 &1 \\
          0 & 0 & 0 & 0 & 0 & 1 & 0 & -1\\
          0 &0 & 0& 0 & 0 & 0 & 1 & 1 \\
        \end{array}
      \right)
\end{equation*}

\begin{eqnarray*}
  S_a^{(5)}&=&\left(
        \begin{array}{cccccccccccccccc}
          0 & 1 & 0 & -1 & 0 & -1 & 0 & 0 &0&-1&0&0& 0 & 0 & 0& 0  \\
          0 & 0 & 1 & 1 & 0 & 0 & -1 & 0  &0&0&-1&0& 0 & 0 & 0& 0\\
          0 & 0 & 0 & 0 & 1 & 1 & 1 & 0&0&0&0&0& -1 & 0 & 0& 0 \\
          0 & 0 & 0& 0 & 0 & 0 & 0 & 1 &0&0&0&0& 0 & 0 & 0& -1\\
          0 & 0 & 0 & 0 &0 & 0 & 0 & 0&1 & 1 & 1 & 0 &1 & 0 & 0 & 0\\
          0 & 0 & 0 & 0 &0 & 0 & 0 & 0&0 & 0 & 0 & 1& 0 & 0 & 0 &1\\
          0 & 0 & 0 & 0 &0 & 0 & 0 & 0 &0 & 0 & 0 & 0 &0 & 1 & 0 & -1\\
          0 & 0 & 0 & 0 &0 & 0 & 0 & 0 &0 & 0 & 0 & 0& 0 & 0 & 1 & 1\\
        \end{array}
      \right),\\
   \tilde{S}_a^{(5)}&=&\left(
        \begin{array}{cccccccccccccccc}
          1 & 1 & 1 & 0 &1 & 0 & 0 & 0 & 1 & 0 & 0 & 0 & 0 & 0 & 0 & 0 \\
          0 & 0 & 0 & 1 &0 & 0 & 0 & 1 & 0 & 0 & 0 & 1 & 0 & 0 & 0 & 0\\
          0 & 0 & 0 & 0 &0 & 1 & 0 & -1 & 0 & 0 & 0 & 0 & 0 & 1 & 0 & 0 \\
          0 & 0 & 0 & 0 &0 & 0 & 1 & 1 & 0 & 0 & 0 & 0 & 0 & 0 & 1 & 0\\
          0 & 0 & 0 & 0 &0 & 0 & 0 & 0 & 0 & 1 & 0 & -1 & 0 & -1 & 0 & 0\\
          0 & 0 & 0 & 0 &0 & 0 & 0 & 0 & 0 & 0 & 1 & 1 & 0 & 0 & -1 & 0 \\
          0 & 0 & 0 & 0 &0 & 0 & 0 & 0 & 0 & 0 & 0 & 0 & 1 & 1 & 1 & 0\\
          0 & 0 & 0 & 0 &0 & 0 & 0 & 0 & 0 & 0 & 0 & 0 & 0 & 0 & 0 & 1\\
        \end{array}
      \right).
\end{eqnarray*}

The second parity node of the $(5,3)$ Zigzag code,  $(6,4)$ Zigzag code, and $(7,5)$ Zigzag code, can be respectively optimally
repaired by the following matrices
\begin{equation*}
  S_b^{(3)}=\left(
        \begin{array}{cccc}
                   0 & 1 & 0 & 1  \\
             0 & 0 & 1 & -1  \\
        \end{array}
      \right),\ \ \tilde{S}_b^{(3)}=\left(
        \begin{array}{cccc}
          1 & -1 & -1 & 0 \\
0 & 0 & 0 & 1  \\
        \end{array}
      \right)
\end{equation*}

\begin{equation*}
  S_b^{(4)}=\left(
        \begin{array}{cccccccc}
          0 & 1 & 0 & 1 & 0 & 1 & 0 & 0 \\
0 & 0 & 1 & -1 & 0 & 0 & 1 & 0  \\
0 & 0 & 0 & 0 & 1 & -1 & -1 & 0  \\
0 & 0 & 0 & 0 & 0 & 0 & 0 & 1 \\
        \end{array}
      \right)
  ,\ \ \tilde{S}_b^{(4)}=\left(
        \begin{array}{cccccccc}
     1 & -1 & -1 & 0 & -1 & 0 & 0 & 0  \\
0 & 0 & 0 & 1 & 0 & 0 & 0 & -1  \\
0 & 0 & 0 & 0 & 0 & 1 & 0 & 1  \\
0 & 0 & 0 & 0 & 0 & 0 & 1 & -1  \\
        \end{array}
      \right)
\end{equation*}

\begin{eqnarray*}
  S_b^{(5)}&=&\left(
        \begin{array}{cccccccccccccccc}
          0 & 1 & 0 & 1 & 0 & 1 & 0 & 0 & 0 & 1 & 0 & 0 & 0 & 0 & 0 & 0  \\
0 & 0 & 1 & -1 & 0 & 0 & 1 & 0 & 0 & 0 & 1 & 0 & 0 & 0 & 0 & 0  \\
0 & 0 & 0 & 0 & 1 & -1 & -1 & 0 & 0 & 0 & 0 & 0 & 1 & 0 & 0 & 0  \\
0 & 0 & 0 & 0 & 0 & 0 & 0 & 1 & 0 & 0 & 0 & 0 & 0 & 0 & 0 & 1  \\
0 & 0 & 0 & 0 & 0 & 0 & 0 & 0 & 1 & -1 & -1 & 0 & -1 & 0 & 0 & 0  \\
0 & 0 & 0 & 0 & 0 & 0 & 0 & 0 & 0 & 0 & 0 & 1 & 0 & 0 & 0 & -1  \\
0 & 0 & 0 & 0 & 0 & 0 & 0 & 0 & 0 & 0 & 0 & 0 & 0 & 1 & 0 & 1  \\
0 & 0 & 0 & 0 & 0 & 0 & 0 & 0 & 0 & 0 & 0 & 0 & 0 & 0 & 1 & -1  \\
        \end{array}
      \right),\\
   \tilde{S}_b^{(5)}&=&\left(
        \begin{array}{cccccccccccccccc}
          1 & -1 & -1 & 0 & -1 & 0 & 0 & 0 & -1 & 0 & 0 & 0 & 0 & 0 & 0 & 0  \\
0 & 0 & 0 & 1 & 0 & 0 & 0 & -1 & 0 & 0 & 0 & -1 & 0 & 0 & 0 & 0  \\
0 & 0 & 0 & 0 & 0 & 1 & 0 & 1 & 0 & 0 & 0 & 0 & 0 & -1 & 0 & 0  \\
0 & 0 & 0 & 0 & 0 & 0 & 1 & -1 & 0 & 0 & 0 & 0 & 0 & 0 & -1 & 0  \\
0 & 0 & 0 & 0 & 0 & 0 & 0 & 0 & 0 & 1 & 0 & 1 & 0 & 1 & 0 & 0  \\
0 & 0 & 0 & 0 & 0 & 0 & 0 & 0 & 0 & 0 & 1 & -1 & 0 & 0 & 1 & 0  \\
0 & 0 & 0 & 0 & 0 & 0 & 0 & 0 & 0 & 0 & 0 & 0 & 1 & -1 & -1 & 0  \\
0 & 0 & 0 & 0 & 0 & 0 & 0 & 0 & 0 & 0 & 0 & 0 & 0 & 0 & 0 & 1  \\
        \end{array}
      \right).
\end{eqnarray*}
\end{Example}

%
%


\end{document}